\documentclass[11pt]{amsart}

\usepackage[centertags]{amsmath}
\usepackage{amsthm,amsfonts}
\usepackage{amssymb}
\usepackage{epsfig}
\usepackage{amssymb,latexsym}
\usepackage{hyperref}
\usepackage{indentfirst}

\usepackage{mathrsfs}
\usepackage{amsthm}
\usepackage{amsmath}
\usepackage{amssymb}
\usepackage{amsfonts}
\usepackage{amsbsy}
\usepackage{tikz}
\usetikzlibrary{matrix,fit,shapes.misc,positioning}
\tikzset{%
  highlight/.style={rectangle,rounded corners,fill=red!15,draw,fill opacity=0.3,thick,inner sep=0pt}
}
\tikzset{%
  highlight1/.style={rectangle,rounded corners,fill=blue!15,draw,fill opacity=0.3,thick,inner sep=0pt}
}

\theoremstyle{plain}
\newtheorem{thm}{Theorem}[section]
\newtheorem{cor}[thm]{Corollary}
\newtheorem{lem}[thm]{Lemma}
\newtheorem{prop}[thm]{Proposition}

\theoremstyle{definition}
\newtheorem{defn}[thm]{Definition}
\newtheorem{rem}[thm]{Remark}
\newtheorem{ex}[thm]{Example}
\numberwithin{equation}{section}
\newcommand{\Fq}{\mathbb{F}_{q}}
\newcommand{\C}{\mathcal{C}}

\newcommand{\E}{\mathbb{E}}
\newcommand{\F}{{\mathbb F}}
\newcommand{\G}{{\mathbb G}}
\newcommand{\HH}{{\mathbb H}}
\newcommand{\Tr}{{\rm Tr}}
\newcommand{\Span}{{\rm Span}}

\begin{document}

\title{Structure and Performance of Generalized Quasi-Cyclic Codes}
\maketitle

\author
{ {\large \begin{center} Cem G\"{u}neri$^{1}$, Ferruh
\"{O}zbudak$^{3}$, Buket \"{O}zkaya$^{1}$,  Elif Sa\c{c}\i kara
$^{1}$,\end{center} \begin{center} Zahra Sepasdar$^{4}$, Patrick
Sol\'e$^{2}$ \end{center}}
\vspace{0.3cm} \small
\begin{center}
$^1$ Sabanc{\i} University, Istanbul, Turkey \\
$^2$ CNRS/LAGA, University of Paris 8, 93 526 Saint-Denis, France,\\
$^3$ Middle East Technical University, Ankara, Turkey\\
$^4$ Department of Pure Mathematics, Ferdowsi University of Mashhad, Iran\\
\end{center}
}


\abstract Generalized quasi-cyclic (GQC) codes form a natural
generalization of quasi-cyclic (QC) codes. They are viewed here as
mixed alphabet codes over a family of ring alphabets. Decomposing
these rings into local rings by the Chinese Remainder Theorem yields
a decomposition of GQC codes into a sum of concatenated codes. This
decomposition leads to a trace formula, a minimum distance bound,
and to a criteria for the GQC code to be self-dual or to be linear
complementary dual (LCD). Explicit long GQC codes that are LCD, but
not QC, are exhibited.
\endabstract

\vspace{0.3cm}


\vspace{0.4cm}

\maketitle

\section{Introduction}

Quasi-cyclic codes (QC) have been known for more than fifty years.
They have been shown to be asymptotically good, which is in marked
contrast with the subclass of cyclic codes. Even in short length
(less than a hundred) they contain more optimal codes than cyclic
codes. Still, their structure is more complex than that of cyclic
codes. Let $q$ denote a prime power and $\F_q$ be the finite field
of that order. A linear code over $\F_q$ is said to be a
quasi-cyclic code of index $\ell$ and length $n=\ell m,$ if and only
if it is held invariant by $T^{\ell}$, where $T$ is the standard
coordinate shift on $n$ digits and $\ell$ is the smallest number
with this property. The approach of \cite{LS} is to view such a code
as mapped from a code of length $\ell$ over the ring
$$R=\F_q[x]/\langle x^m-1\rangle.$$
In recent years a super class of quasi cyclic codes has appeared:
generalized quasi-cyclic codes (\cite{EY,SK}). Up to coordinate
permutation a QC code is equivalent to a linear code with block
circulant generator matrix. More specifically, the circulant blocks
will have the same size, namely the co-index $m.$ The idea of
Generalized Quasi-Cyclic (GQC) codes is to relax this requirement to
allow blocks of different sizes. The immediate benefit is to
construct codes whose lengths are not multiple of the index. For
instance a GQC code might very well have prime length. The CRT
decomposition has been extended to GQC codes at the price of a more
complicated notation (\cite{EY}).

The aim of this paper is twofold. First, we aim to extend the
structural theory of \cite{LS} to GQC codes, a program partially
done in \cite{EY}. In particular the trace formula of \cite{LS} is
extended to GQC codes. Concatenated description of GQC codes is
presented and the results for QC codes in \cite{GO} is extended to
GQC codes. Moreover, multilevel (generalized) concatenated
description of GQC codes is obtained, which yields a minimum
distance bound for GQC codes, extending Jensen's bound for QC codes.
Let us note that there is a minimum distance bound on GQC codes due
to Esmaeili and Yari (\cite{EY}) but it only applies to one
generator GQC codes. Our bound applies to all GQC codes. Criteria
for self-duality bearing on the component codes are given. In a
recent paper \cite{GOS}, a similar criterion for a QC code to
intersect its dual trivially (LCD code as Linear Complementary Dual)
was derived. This criterion is generalized here to GQC codes.

Next, we study the asymptotic performance of GQC codes. Explicit
long GQC codes that are LCD, but not QC, are exhibited. These codes
have only finitely many distinct co-indices in the spirit of
\cite{LS3}, but have a length going to infinity. The proof rests on
the existence of families of good QC codes that are LCD \cite{GOS}.

The material is organized as follows. The next section collects the
necessary definitions and notations. Section \ref{trace section}
develops the concatenated structure and a trace expression. Section
\ref{multilevel} derives the minimum distance bound. Section
\ref{sd-lcd section} derives criteria for self-duality and LCDness.
Asymptotic results are given in Section \ref{asymptotics section}.
Section \ref{conclusion} concludes the article and points out
directions for future research.

\section{Background on QC and GQC codes}
\label{intro} Let $\Fq$ denote the finite field with $q$ elements,
where $q$ is a prime power. A linear code $C$ of length $m\ell$ over
$\Fq$ is called a quasi-cyclic (QC) code of index $\ell$ if it is
invariant under shift of codewords by $\ell$ positions and $\ell$ is
the minimal number with this property. Note that if $\ell=1$, then
$C$ is a cyclic code. If we view codewords of $C$ as $m \times \ell$
arrays as follows
\begin{equation}\label{array}c=\left(
  \begin{array}{ccc}
    c_{00} & \ldots & c_{0,\ell-1} \\
    \vdots &  & \vdots \\
    c_{m-1,0} & \ldots & c_{m-1,\ell-1} \\
  \end{array}
\right),\end{equation} then being invariant under shift by $\ell$
units amounts to being closed under row shift.

Let us define the quotient ring $R:=\Fq[x]/\langle x^m-1 \rangle$.
To an element $c\in \Fq^{m\times \ell} \simeq \Fq^{m\ell}$ as in
(\ref{array}), we associate an element of $R^\ell$
\begin{equation} \label{associate-1}
\vec{c}(x):=(c_0(x),c_1(x),\ldots ,c_{\ell-1}(x)) \in R^\ell ,
\end{equation}
where for each $0\leq j \leq \ell-1$, \begin{equation}
\label{columns} c_j(x):= c_{0,j}+c_{1,j}x+c_{2,j}x^2+\cdots +
c_{m-1,j}x^{m-1} \in R .\end{equation} Then, the following map is an
$\F_q$-linear isomorphism.
\begin{equation}\begin{array}{lll} \label{identification-1}
\phi: \hspace{2cm} \F_q^{m\ell} & \longrightarrow & R^\ell  \\
c=\left(
  \begin{array}{ccc}
    c_{00} & \ldots & c_{0,\ell-1} \\
    \vdots &  & \vdots \\
    c_{m-1,0} & \ldots & c_{m-1,\ell-1} \\
  \end{array}
\right) & \longmapsto & \vec{c}(x) .
\end{array}\end{equation}
Note that for $\ell=1$, this amounts to the classical polynomial
representation of cyclic codes. Observe that $\ell$ shift on
$\F_q^{m\ell}$ corresponds to componentwise multiplication by $x$ in
$R^\ell$ and a $q$-ary QC code $C$ of length $m\ell$ and index
$\ell$ can be considered as an $R$-submodule in $R^\ell$.

Let us now recall the decomposition of a length $m\ell$ QC code over
$\F_q$ into shorter codes over extensions of $\F_q$. We follow the
brief presentation in \cite{GO} and refer the reader to \cite{LS}
for details. We assume that $\gcd(m,q)=1$ and factor the polynomial
$x^m-1$ into pairwise distinct irreducible polynomials in $\F_q[x]$
as
\begin{equation}\label{irreducibles}
x^m-1=f_1(x)f_2(x)\cdots f_s(x).
\end{equation}
By Chinese Remainder Theorem, we have the following ring
isomorphism:
\begin{equation} \label{CRT-1}
R\cong \bigoplus_{i=1}^{s} \F_q[x]/\langle f_i(x)\rangle .
\end{equation}
Since each $f_i(x)$ divides $x^m-1$, their roots are powers of some
fixed primitive $m^{th}$ root of unity $\xi$. For each
$i=1,\ldots,s$, let $u_i$ be the smallest nonnegative integer such
that $f_i(\xi^{u_i})=0$. Since $f_i(x)$'s are irreducible, direct
summands in (\ref{CRT-1}) are field extensions of $\F_q$. If
$\E_i:=\F_q[x]/\langle f_i(x) \rangle$ for $1\leq i \leq s$, then we
have
\begin{equation} \label{CRT-2}
R^{\ell}\cong \E_1^{\ell} \oplus \cdots  \oplus \E_{s}^{\ell}.
\end{equation}
Hence, a QC code $\C\subset R^\ell$ can be viewed as an $(\E_1
\oplus \cdots \oplus \E_{s})$-submodule of $\E_1^{\ell} \oplus
\cdots  \oplus \E_{s}^{\ell}$ and decomposes as
\begin{equation} \label{constituents}
C=C_1\oplus \cdots  \oplus C_{s},
\end{equation}
where $C_i$ is a linear code of length $\ell$ over $\E_i$, for each
$i$. These length $\ell$ linear codes over various extensions of
$\F_q$ are called the constituents of $C$.

If $C\subset R^\ell$ is generated as an $R$-module by
$$\{\bigl(a_{0}^1(x),\ldots ,a_{\ell-1}^1(x)\bigr),\ldots ,
\bigl(a_{0}^r(x),\ldots ,a_{\ell-1}^r(x)\bigr)\} \subset R^\ell ,$$
then
\begin{equation}\label{explicit constituents}
C_i  =  \Span_{\E_i}\bigl\{\bigl(a_{0}^b(\xi^{u_i}),\ldots
,a_{\ell-1}^b(\xi^{u_i})\bigr): 1\leq b \leq r \bigr\}, \ \mbox{for
$1\leq i \leq s$}.
\end{equation}

Another way of decomposing QC codes is given by Jensen (\cite{J}) by
the concatenation method. For each $1\leq i \leq s$, consider the
minimal cyclic code of length $m$ over $\F_q$, whose check
polynomial is $f_i(x)$. Let $\theta_i$ denote the generating
primitive idempotent for each minimal cyclic code in consideration.
Jensen showed the following.

\begin{thm} \cite{J} \label{Jensen's thm}
(i) Let $C$ be a length $m\ell$ and index $\ell$ QC code over $\Fq$.
Then there exist linear codes $\mathfrak{C}_i$ of length $\ell$ over
$\E_i$ such that $C=\displaystyle\bigoplus_{i=1}^s \langle \theta_i
\rangle \Box \mathfrak{C}_i$.

(ii) Conversely, let $\mathfrak{C}_i$ be an $\E_i$-linear code of
length $\ell$ for each $i\in \{1,\ldots ,s\}$. Then,
$C=\displaystyle\bigoplus_{i=1}^s  \langle \theta_i \rangle \Box
\mathfrak{C}_i$ is a $q$-ary QC code of length $m\ell$ and index
$\ell$.
\end{thm}

Note that each field $\E_i$ is isomorphic to $\langle \theta_i
\rangle$, for each $1\leq i \leq s$, via the maps
\begin{eqnarray} \label{isoms}
\begin{array}{ccc} \varphi_i:\langle \theta_i \rangle
& \longrightarrow & \E_i \\ \hspace{0.5cm} a(x)& \longmapsto &
a(\xi^{u_i}) \end{array}
& & \begin{array}{ccc} \psi_i: \E_i & \longrightarrow & \langle \theta_i \rangle \\
\hspace{0.5cm} \delta & \longmapsto & \sum\limits_{k=0}^{m-1} a_kx^k
\end{array}\ \ ,
\end{eqnarray}
where
$$a_k=\frac{1}{m} \Tr_{\E_i/\F_q}(\delta\xi^{-ku_i}).$$
It is easy to observe that $\varphi_i$ and $\psi_i$ are inverse to
each other. Let us note that for each $i$, the concatenation of the
minimal cyclic code $\langle \theta_i \rangle$ and the linear code
$\mathfrak{C}_i$ over $\E_i$ is carried out by the map $\psi_i$,
which identifies the field $\E_i$ with the minimal cyclic code. In
other words, a codeword $(c_0,\ldots,c_{\ell-1})$ in some
$\mathfrak{C}_i$ is mapped to
$\left(\psi_i(c_0),\ldots,\psi_i(c_{\ell-1})\right)$ in $R^{\ell}$.

It is proved in \cite{GO} that for a given QC code $C$, the
constituents $C_i$'s in (\ref{constituents}) and the outer codes
$\mathfrak{C}_i$'s in the concatenated structure are equal to each
other (see \cite[Theorem 4.1]{GO} ).

By (\ref{isoms}), the concatenated structure of QC codes can be used
to demonstrate the trace representation of QC codes given by
Ling-Sol\'{e}, which provides a vectorial representation of
codewords equivalent to (\ref{array}), when the constituents are
known.

\begin{thm} \cite[Theorem 5.1]{LS} \cite[Theorem 4.2]{GO} \label{traces}
Consider the QC code $C$ with the constituents $C=C_{1}\oplus \cdots
\oplus C_{s},$ where $C_i \subset
\E_{i}^\ell=\F_q(\xi^{u_{i}})^\ell$ is linear over $\E_{i}$ of
length $\ell$ for each $1\leq i \leq s$. Then an arbitrary codeword
$c\in C$ as an $m\times \ell$ array has the form
$$c=\left(\begin{array}{c}c_0(\lambda_1,\ldots
,\lambda_{s}) \\ c_1(\lambda_1,\ldots ,\lambda_{s}) \\ \vdots \\
c_{m-1}(\lambda_1,\ldots ,\lambda_{s})
\end{array} \right),$$
where $\lambda_i=(\lambda_{i,0},\ldots ,\lambda_{i,\ell-1})$ is a
codeword in $C_i$ for each $i$ and $$c_k(\lambda_1,\ldots
,\lambda_{s})=\left(\sum\limits_{i=1}^{s}\Tr_{\E_{i}/\F_q}\left(\lambda_{i,j}\xi^{-ku_{i}}
\right) \right)_{0\leq j \leq \ell-1} ,$$ for each $0\leq k \leq
m-1$.
\end{thm}
\noindent If we set $\ell=1$ above, then we get the trace
representation of a $q$-ary cyclic code of length $m$.

Generalized quasi-cyclic (GQC) codes were introduced in \cite{SK},
where their description is given as follows.

\begin{defn}
Let $m_0,\ldots,m_{\ell-1}$ be positive integers and set
$R_j:=\Fq[x]/\langle x^{m_j}-1\rangle$ for each $j=0,\ldots,\ell-1$.
An $\Fq[x]$-submodule of $R':=R_0 \times \cdots \times R_{\ell-1}$
is called a generalized quasi-cyclic (GQC) code of block lengths
$(m_0,\ldots,m_{\ell-1})$, which is a linear code of length $m_0
+\cdots +m_{\ell-1}$ over $\Fq$.
\end{defn}
Note that if $m_0=\cdots = m_{\ell-1}=m$, then we obtain a
quasi-cyclic code of length $m\ell$ and index $\ell$.

\begin{ex}
For any finite field $\F_q$ and any positive integer $n$, the
$q$-ary repetition code of length $n$ is a GQCCD code for any
partition of $n$. Its dual, namely the parity check code of length
$n$ is also GQCCD. A nontrivial example is the class of binary
Cordaro-Wagner codes(\cite{CW}), which are, by definition, two
dimensional codes attaining the best possible distance. Given an
$[n,2,d]$ Cordaro-Wagner code with the column partition $(h,j,k)$,
each $h, j$ and $k$ gives the number of the nonzero binary columns,
namely $\small\left[\begin{array}{c} 1 \\ 0\end{array}\right]$, $\small\left[\begin{array}{c} 0 \\
1\end{array}\right]$ and $\small\left[\begin{array}{c} 1 \\
1\end{array}\right]$ \normalsize in the generator matrix. Note that
this code is self-orthogonal if $n$ is multiple of 6 and
complementary-dual otherwise (see Table I in \cite{CW}).

Consider the binary $[16, 2, 10]$ Cordaro-Wagner code $C$ generated
by
$$G=\left(
      \begin{array}{cccccccccccccccc}
        1 & 1 & 1 & 1 & 1 & 1 & 0 & 0 & 0 & 0 & 0 & 1 & 1 & 1 & 1 & 1 \\
        0 & 0 & 0 & 0 & 0 & 0 & 1 & 1 & 1 & 1 & 1 & 1 & 1 & 1 & 1 & 1 \\
      \end{array}
    \right),
$$
where $h=6$, $j=k=5$.

If we set $m_1=h$, $m_2=j$, $m_3=k$ and $\ell=3$, then $C$ is a
binary GQC code with:
$$C=\langle(x^5+x^4+x^3+x^2+x+1, 0, x^4+x^3+x^2+x+1), (0, x^4+x^3+x^2+x+1, x^4+x^3+x^2+x+1)\rangle.$$
It is easy to observe that any Cordaro-Wagner code will give a
binary 2-generator GQC code, where the polynomial coordinates of the
generators will be either $0$ or $\frac{x^m-1}{x-1}$ for $m=h,j,k$.
\end{ex}

The factorization of GQC codes into constituents is given by
Esmaeili and Yari in \cite{EY}. We will review this decomposition
and introduce a notation which is suitable for presentation of our
results in the rest of the article.

Let $\gcd(m_j,q)=1$ for each $j=0,\ldots,\ell-1$, then each
$x^{m_j}-1$ factors into distinct irreducible polynomials. Suppose
that the total number of distinct irreducible factors over all
$x^{m_j}-1$ decompositions is $s$ and let $f_{1}(x),\ldots
,f_{s}(x)$ denote these irreducible polynomials. Then for each $j$
we have
\begin{equation}\label{irreducibles-2}
x^{m_j}-1=f_{1}(x)^{v_{1,j}}f_{2}(x)^{v_{2,j}}\cdots
f_{s}(x)^{v_{s,j}},
\end{equation}
where $v_{i,j} \in \{0,1\}$. Since $f_{i}(x)$'s are irreducible,
$\F_q[x]/\langle f_{i}(x)\rangle$ is a finite field extension of
$\F_q$. Set $\E_i:=\F_q[x]/\langle f_{i}(x)\rangle$ for $1\leq i\leq
s$ and for $1\leq i\leq s$, $0\leq j \leq \ell -1$, define
\begin{equation} \label{extensions}
\E_{i,j}:= \begin{cases}
   \ \E_i,   \text{ if } v_{i,j}=1, \\
   \{0\},   \text{ if } v_{i,j}=0.
  \end{cases}.
\end{equation}

Let us fix a root $\alpha_i$ of each $f_i$ ($1\leq i \leq s$). For
$a(x)\in R_j$ and $1\leq i \leq s$, set
\begin{equation}\label{evaluation}
a_{i,j}= \begin{cases}
   a(\alpha_i),   \text{ if } \E_{i,j}=\E_i, \\
   0,   \text{ \hspace{0.53cm} if } \E_{i,j}=\{0\}.
  \end{cases}
\end{equation}
By (\ref{irreducibles-2}), (\ref{extensions}) and the Chinese
Remainder Theorem, we get the following ring isomorphism for each
$j=0,\ldots,\ell-1$:
\begin{equation} \label{CRT-4}
R_j\cong \bigoplus_{i=1}^{s} \E_{i,j} ,
\end{equation}
where the isomorphism maps $a(x)\in R_j$ to $(a_{1,j}+\cdots
+a_{s,j})$ (cf. (\ref{evaluation})). Therefore we have
\begin{equation} \label{CRT-6}
R'= R_0 \times \cdots \times R_{\ell-1}\cong
\left(\bigoplus_{i=1}^{s} \E_{i,0} \right)\times \cdots \times
\left(\bigoplus_{i=1}^{s} \E_{i,\ell-1} \right) \cong
\bigoplus_{i=1}^{s} (\E_{i,0} \times \cdots \times \E_{i,\ell-1}),
\end{equation}
where $\left(a^{0}(x),\ldots ,a^{\ell-1}(x)\right)\in R'$ is mapped
to $\displaystyle{\sum_{i=1}^{s}
\left(a^{0}_{i,0},a^{1}_{i,1},\ldots , a^{\ell -1}_{i,\ell -1}
\right)}$. In particular, a GQC code $C\subset R'$ can be viewed
inside $\displaystyle{\bigoplus_{i=1}^s \E_i^{\ell}}$ since for each
$j$, $\E_{i,j}$ is either $\E_i$ or $\{0\}\subset\E_i$.

\begin{prop}\label{GQCconstituents} Suppose the GQC code $C\subset R'$ is generated as an $\F_q[x]$-module by
\begin{equation*} \left\{\bigl(a^{1,0}(x),\ldots ,a^{1,\ell -1}(x)\bigr),\ldots ,
\bigl(a^{r,0}(x),\ldots ,a^{r,\ell -1}(x)\bigr)\right\} \subset
R'.\end{equation*} Then $C$, as a subset of
$\displaystyle{\bigoplus_{i=1}^s \E_i^{\ell}}$, can be written as
\begin{equation} \label{constituents-2}
C=C_1\oplus \cdots  \oplus C_{s},
\end{equation}
where each $C_i$ (constituent) is an $\E_i$-linear code of length
$\ell$ and described as
\begin{equation}\label{explicit constituents-2}
C_i  = \Span_{\E_i}\left\{\left(a_{i,0}^{b,0},\ldots
,a_{i,\ell-1}^{b,\ell -1}\right): 1\leq b \leq r \right\}, \
\mbox{for $1\leq i \leq s$}.
\end{equation}
\end{prop}

\begin{proof}
Observe that $C$, as a subset of $R'$, can be written as
$$C=\left\{g_1(x)\left(a^{1,0}(x),\ldots ,a^{1,\ell -1}(x) \right)+ \cdots + g_r(x)\left(a^{r,0}(x),\ldots ,a^{r,\ell -1}(x) \right): \ g_1,\ldots ,g_r \in \F_q[x]  \right\}.$$
Then by (\ref{CRT-6}), $C_i$ is of the form
$$C_i=\left\{g_1(\alpha_i)\left(a_{i,0}^{1,0},\ldots
,a_{i,\ell-1}^{1,\ell -1} \right) + \cdots +
g_r(\alpha_i)\left(a_{i,0}^{r,0},\ldots ,a_{i,\ell-1}^{r,\ell -1}
\right): \ g_1,\ldots ,g_r \in \F_q[x] \right\}.$$ Since $\alpha_i$
is a root of $f_i(x)$, we have $\E_i=\F_q(\alpha_i)$. Therefore the
elements $g_1(\alpha_i),\ldots ,g_r(\alpha_i)$ take all possible
values in $\E_i$ as the polynomials $g_1,\ldots ,g_r$ range over
$\F_q[x]$. Hence the result follows.
\end{proof}

\begin{rem}
Depending on $v_{i,j}$'s in the factorization
(\ref{irreducibles-2}), some $\E_{i,j}$'s can be $\{0\}$ and hence
corresponding coordinates of all the codewords in the related
constituent will be 0.
\end{rem}

\begin{ex}
Let $q=2$, $m_0=3$, $m_1=5$, $m_2=9$ and hence $\ell=3$. We have
$$R'=R_0 \times R_1 \times R_2 = \F_2[x]/\langle x^3-1\rangle \times
\F_2[x]/\langle x^5-1\rangle \times \F_2[x]/\langle x^9-1\rangle$$
and
\begin{eqnarray*}
x^3-1 &=& (x+1)(x^2+x+1),\\
x^5-1 &=& (x+1)(x^4+x^3+x^2+x+1),\\
x^9-1 &=& (x+1)(x^2+x+1)(x^6+x^3+1).
\end{eqnarray*}
Let $f_1(x)=x+1$, $f_2(x)=x^2+x+1$, $f_3(x)=x^4+x^3+x^2+x+1$ and
$f_4(x)=x^6+x^3+1$. Then we have $\E_1\simeq \F_2$, $\E_2\simeq
\F_4$, $\E_3\simeq \F_{16}$ and $\E_4\simeq \F_{64}$. Moreover, with
the notation in (\ref{extensions}), we have the following:
\[\begin{array}{llll} \E_{1,0}=\E_1 \ & \E_{2,0}=\E_2 \ & \E_{3,0}=\{0\} \ & \E_{4,0}=\{0\} \ \\
\E_{1,1}=\E_1 \ & \E_{2,1}=\{0\} \ & \E_{3,1}=\E_3 \ & \E_{4,1}=\{0\} \ \\
\E_{1,2}=\E_1 \ & \E_{2,2}=\E_2 \ & \E_{3,2}=\{0\} \ & \E_{4,2}=\E_4 \ \\
\end{array} \]
Hence,
$$R'\simeq (\E_1 \times \E_1 \times \E_1) \oplus (\E_2 \times \{0\} \times \E_2) \oplus (\{0\}\times \E_3 \times \{0\}) \oplus (\{0\}\times \{0\}\times \E_4).$$
Let us fix roots of $f_1,\ldots, f_4$ as $\alpha_1=1,
\alpha_2,\alpha_3,\alpha_4$. If $C \subset R'$ is a GQC code
generated by $\langle \vec{g}_1(x),\ldots,\vec{g}_r(x)\rangle$,
where
$$\vec{g}_b(x)=\left(g^{b,0}(x),g^{b,1}(x),g^{b,2}(x) \right), \ 1\leq b \leq r,$$
then $C$ has the following constituents:
\begin{eqnarray*}
C_1 &=& \Span_{\F_2}\left\{\left(g^{b,0}(1),g^{b,1}(1),g^{b,2}(1)\right)\ :\ 1\leq b \leq r \right\},\\
C_2 &=& \Span_{\F_4}\left\{\left(g^{b,0}(\alpha_2),0,g^{b,2}(\alpha_2)\right)\ :\ 1\leq b \leq r \right\},\\
C_3 &=& \Span_{\F_{16}}\left\{\left(0,g^{b,1}(\alpha_3),0 \right)\ :\ 1\leq b \leq r \right\},\\
C_4 &=& \Span_{\F_{64}}\left\{\left(0,0,g^{b,2}(\alpha_4)\right)\ :\
1\leq b \leq r \right\}.
\end{eqnarray*}
\end{ex}

\section{Concatenated Structure and Trace Representation} \label{trace section}
Our goal is to obtain, as in the QC codes, a concatenated
description and its relation to constituent decomposition for GQC
codes. For this purpose, some further notation needs to be
introduced. We will also continue using the notation of the previous
section.

For $i,j$ such that $f_i(x) \mid x^{m_j}-1$, let $\theta_{i,j}$
denote the primitive idempotent generator of the minimal cyclic code
of length $m_j$ in $R_j$, whose check polynomial is $f_i(x)$. Let
$0_{j}$ denote the zero codeword of length $m_j$ (or the zero
polynomial in $R_j$). Then define the following polynomials for each
$i$ and $j$:
\begin{equation} \label{idempotent}
I_{i,j} :=
\begin{cases}
   \theta_{i,j}(x),       & \text{if } f_i(x)\mid x^{m_j}-1, \\
   0_{j}       & \text{otherwise.}
  \end{cases}
\end{equation}

Now we can define the following analogues of the maps in
(\ref{isoms}) for each block of length $m_j$ and each $1\leq i\leq
s$:
\begin{eqnarray} \label{isoms-2}
\begin{array}{cccc} \varphi_{i,j}:& \langle I_{i,j} \rangle
& \longrightarrow & \hspace{0.5cm} \E_{i,j} \\ \hspace{0.5cm}  &
a(x)& \longmapsto & a_{i,j}
\end{array}
& \enskip & \begin{array}{cccc} \psi_{i,j}: & \E_{i,j}  & \longrightarrow & \langle I_{i,j} \rangle \\
\hspace{0.5cm} & \delta & \longmapsto & \sum\limits_{k_j=0}^{m_j-1}
a_{k_j}x^{k_j}
\end{array}\ \ ,
\end{eqnarray}
where
$$a_{k_j}=\frac{1}{m_j} \Tr_{\E_i/\F_q}(\delta \alpha_i^{-k_j}).$$
Note that $\langle I_{i,j} \rangle = \langle 0_{j} \rangle$,
$\E_{i,j}= \{0\}$ and $a_{i,j}=0$ are equivalent and all amount to
$f_i(x) \nmid x^{m_j}-1$. Then, $\varphi_{i,j}$ and $\psi_{i,j}$ are
well-defined $\E_i$-linear isomorphisms and they are inverses to
each other for all $i$ and $j$. Moreover, when $\E_{i,j}=\E_i$,
hence $I_{i,j}=\theta_{i,j}$, $\psi_{i,j}$ and $\phi_{i,j}$ are
known to be field isomorphisms. In particular, if
$m_0=\cdots=m_{\ell-1}$, then we obtain the isomorphisms in
(\ref{isoms}) for the QC case.

Note that $R'=R_0\times \cdots \times R_{\ell -1}$ and
$\E_{i,0}\times \cdots \times \E_{i,\ell -1}$ (for each $1\leq i
\leq s$) are rings with coordinate-wise addition and multiplication.
The multiplicative identity of $R'$ is clearly $1_{R'}:=(1,\ldots
,1)$. For all $1\leq i \leq s$, $0\leq j \leq \ell -1$, set
\begin{equation}\label{one}
1_{i,j}:= \begin{cases}
   1_{\E_i},   \text{ \hspace{0.65cm} if } \E_{i,j}=\E_i, \\
   0,   \text{ \hspace{0.95cm} if } \E_{i,j}=\{0\}.
  \end{cases}
\end{equation}
Then, $1_i:=(1_{i,0},\ldots ,1_{i,\ell -1})$ is the multiplicative
identity of $\E_{i,0}\times \cdots \times \E_{i,\ell -1}$ for each
$1\leq i \leq s$. Note also that $\psi_{i,j}(1_{i,j})=I_{i,j}$ for
all $i,j$.

For $i=1,\ldots , s$, we now define two other maps (cf.
(\ref{evaluation}) and (\ref{isoms-2})).
\begin{eqnarray} \label{isoms-3} \begin{array}{ccccc}
\Phi_i& : & R_0 \times \cdots \times R_{\ell -1} & \longrightarrow &
\E_{i,0}\times \cdots \times \E_{i,\ell -1} \\ & & \left(a^0(x),
\ldots, a^{\ell -1}(x)\right)& \longmapsto & \left(a^0_{i,0},\ldots
, a^{\ell -1}_{i,\ell -1}\right)
\end{array} \\
\begin{array}{ccccc} \label{isoms-4}
\Psi_i& : & \E_{i,0}\times \cdots \times \E_{i,\ell -1} &
\longrightarrow & R_0 \times \cdots \times R_{\ell -1} \\ & &
\left(\delta_0,\ldots , \delta_{\ell -1}\right)& \longmapsto &
\left(\psi_{i,0}(\delta_0),\ldots , \psi_{i,\ell -1}(\delta_{\ell
-1}) \right)
\end{array}
\end{eqnarray}
Note that for each $i$, $\Phi_i$ and $\Psi_i$ are $\F_q$-linear maps
and they are also ring homomorphisms. Moreover, when $\Phi_i$ is
restricted to $\langle I_{i,0}\rangle \times \cdots \times \langle
I_{i,\ell -1}\rangle$, they are inverse to each other (cf.
(\ref{isoms-2})). For $i=1,\ldots , s$, we set
$I_i:=\left(I_{i,0},\ldots ,I_{i,\ell -1} \right) \in R'$. We have
$\Psi_i \left(1_i \right)=I_i$ and the ideal generated by $I_i$ in
$R'$ is nothing but $\langle I_{i,0}\rangle \times \cdots \times
\langle I_{i,\ell -1}\rangle$. The next result follows immediately
from the definition of $I_i$'s and the analogous results on
primitive idempotents of cyclic codes (cf. \cite[Theorem
6.4.4]{vL}). Recall that the multiplication and addition in $R'$ are
coordinate-wise.

\begin{lem}\label{idemp properties} The following identities hold in $R'$:
\begin{itemize}
\item[i.] $I_i \cdot I_i = I_i$, for all $i=1,\ldots ,s$.
\item[ii.] $I_u \cdot I_v = 0$, if $u\not= v$.
\item[iii.] $I_1+\cdots + I_s =1_{R'}$.
\end{itemize}
\end{lem}

The next result will be used in proving the concatenated structure
of GQC codes.
\begin{thm}\label{mimicjensen}
With the notation above, we have
$$R'=\bigoplus_{i=1}^s \langle I_i \rangle.$$
\end{thm}

\begin{proof}
We first show that the sum is direct in $R'$. Let $(g_0(x),\ldots
,g_{\ell -1}(x))$ be an element of $\langle I_u \rangle \cap \langle
I_v \rangle$ for some $u\not=v \in \{1,\ldots ,s\}$. Since $\langle
I_u \rangle = \langle I_{u,0}\rangle \times \cdots \times \langle
I_{u,\ell -1}\rangle$ and  $\langle I_v \rangle = \langle
I_{v,0}\rangle \times \cdots \times \langle I_{v,\ell -1}\rangle$,
we have $g_j(x)\in \langle I_{u,j}\rangle \cap \langle
I_{v,j}\rangle$ for all $0\leq j \leq \ell -1$. If one of the
irreducible polynomials $f_u(x)$ or $f_v(x)$ does not divide
$x^{m_j}-1$, say $f_u$, then $\langle I_{u,j}\rangle=\langle 0_j
\rangle$. Therefore $g_j(x)=0$ in this case. If both $f_u(x)$,
$f_v(x)$ divide $x^{m_j}-1$, then $\langle I_{u,j}\rangle$
(respectively, $\langle I_{v,j}\rangle$) is the minimal cyclic code
generated by $(x^{m_j}-1)/f_u(x)$ (respectively,
$(x^{m_j}-1)/f_v(x)$). Since these minimal cyclic codes intersect
trivially, we have $g_j(x)=0$ in this case too. Hence,
$(g_0(x),\ldots ,g_{\ell -1}(x))=(0,\ldots , 0)$ and the sum is
direct.

Clearly $\langle I_i \rangle \subset R'$ for each $i$. Recall that
when $\Phi_i$ is restricted to $\langle I_i \rangle$, $\Phi_i$ and
$\Psi_i$ are inverse $\F_q$-linear maps. Hence,
$\Psi_i\left(\E_{i,0}\times \cdots \times \E_{i,\ell -1} \right)=
\langle I_i\rangle$ and
$$\dim_{\F_q} \langle I_i\rangle=\dim_{\F_q} \left( \E_{i,0}\times \cdots \times \E_{i,\ell -1} \right)=\sum_{\footnotesize{\begin{array}{c}0\leq j \leq \ell -1 \\ f_i(x)\mid (x^{m_j}-1) \end{array}}} \deg f_i,$$ for all $1\leq i \leq s$. Then,
\begin{eqnarray*}
\dim_{\F_q} \bigoplus_{i=1}^s \langle I_i \rangle & = & \sum_{i=1}^{s} \sum_{\footnotesize{\begin{array}{c}0\leq j \leq \ell -1 \\ f_i(x)\mid (x^{m_j}-1) \end{array}}} \deg f_i\\
& = & \sum_{j=0}^{\ell -1} \sum_{\footnotesize{\begin{array}{c}1\leq
i \leq s \\ f_i(x)\mid (x^{m_j}-1) \end{array}}} \deg f_i .
\end{eqnarray*}
For each $0\leq j \leq \ell -1$, we have
$$\sum_{\footnotesize{\begin{array}{c}1\leq i \leq s \\ f_i(x)\mid (x^{m_j}-1) \end{array}}} \deg f_i =m_j,$$
since $\gcd(q,m_j)=1$ and hence $x^{m_j}-1$ is separable. Therefore
$$\dim_{\F_q} \bigoplus_{i=1}^s \langle I_i \rangle=\sum_{j=0}^{\ell -1} m_j.$$
Note that $(m_0+m_1+\cdots + m_{\ell -1})$ is also the
$\F_q$-dimension of $R'$ and therefore the result follows.
\end{proof}



\begin{rem}\label{concat-defn}
For any $i\in \{1,\ldots ,s\}$ and an $\E_i$-linear code
$\mathfrak{C}_i \subset \E_{i,0}\times \cdots \times \E_{i,\ell -1}$
of length $\ell$, concatenation with $\langle I_i \rangle = \langle
I_{i,0} \rangle \times \cdots \times \langle I_{i,\ell -1} \rangle
\subset R'$ is carried out by the map $\Psi_i$ in (\ref{isoms-4}).
Namely,
$$\langle I_i \rangle \Box \mathfrak{C}_i := \left\{\left(\psi_{i,0}(c_0),\ldots,\psi_{i,\ell-1}(c_{\ell-1})\right) :
(c_0,\ldots,c_{\ell-1}) \in \mathfrak{C}_i \right\}.$$
\end{rem}

After this preparation, we can now generalize Theorem \ref{Jensen's
thm} for a GQC code $C \subset R'$ of length $m_0 +\cdots
+m_{\ell-1}$ over $\Fq$.

\begin{thm} \label{generalized Jensen}
(i) Let $C \subset R'$ be a GQC code and $\tilde{C}_i:=C\cdot I_i
\subset R'$ for each $1\leq i \leq s$. Then,
$$C=\bigoplus_{i=1}^s \tilde{C}_i.$$
Moreover, for the $\E_i$-linear code
$\mathfrak{C}_i:=\Phi_i(\tilde{C}_i) \subset \E_{i,0}\times \cdots
\times \E_{i,\ell -1}$ of length $\ell$, we have
$\tilde{C}_i=\langle I_i \rangle \Box \mathfrak{C}_i$ (for all $i$),
so that
$$C=\bigoplus_{i=1}^s \langle I_i \rangle \Box \mathfrak{C}_i.$$

(ii) Conversely, let $\mathfrak{C}_i \subseteq( \E_{i,0} \times
\cdots \times \E_{i,\ell-1})$ be an $\E_i$-linear code of length
$\ell$ for each $i\in \{1,\ldots ,s\}$. Then,
$C=\displaystyle\bigoplus_{i=1}^s  \langle I_{i} \rangle \Box
\mathfrak{C}_i$ is a $q$-ary GQC code of length
$m_0+\cdots+m_{\ell-1}$.
\end{thm}

\begin{proof}
(i) By Lemma \ref{idemp properties}, we have
$$C=C\cdot 1_{R'}=C\cdot \sum_{i=1}^s I_i=\sum_{i=1}^s \tilde{C}_i.$$
Since $\tilde{C}_i \subset \langle I_i \rangle$ for each $i$ and
$\langle I_i \rangle$'s are pairwise intersecting trivially (Theorem
\ref{mimicjensen}), we conclude that $C=\oplus_i \tilde{C}_i$.

We have
\begin{eqnarray*}
\tilde{C}_i & = & \left\{\left(c^0(x),\ldots ,c^{\ell -1}(x)\right)\cdot \left(I_{i,0}(x),\ldots ,I_{i,\ell -1}(x)\right): \  \left(c^0(x),\ldots ,c^{\ell -1}(x)\right)\in C \right\}\\
& = & \left\{\left(c^0(x)I_{i,0}(x),\ldots ,c^{\ell -1}(x)I_{i,\ell
-1}(x)\right): \  \left(c^0(x),\ldots ,c^{\ell -1}(x) \right)\in C
\right\} \subset \langle I_i \rangle.
\end{eqnarray*}
Since $\Phi_i$ restricted to $\langle I_i \rangle$ is an isomorphism
((\ref{isoms-3}) and (\ref{isoms-4})), the last expression is equal
to
$$\left\{\left(\psi_{i,0}\left(d^0_{i,0}\right),\ldots ,\psi_{i,\ell -1}\left(d^{\ell -1}_{i,\ell -1}\right) \right): \ \left(d^0(x),\ldots , d^{\ell -1}(x) \right)\in \Phi_i(\tilde{C}_i) \right\},$$
which is nothing but $\langle I_i \rangle \Box \Phi_i(\tilde{C}_i)$
(cf. (Remark \ref{concat-defn}).

(ii) The concatenation has the form
\begin{equation*} \label{converse}
\langle I_i \rangle \Box \mathfrak{C}_i  = \left\{ \left(
\psi_{i,0}(c_0),\ldots, \psi_{i,\ell-1}(c_{\ell-1})\right): \
(c_0,\ldots ,c_{\ell-1})\in \mathfrak{C}_i \right\} .
\end{equation*}
Note that each $\psi_{i,j}(c_j)$ is an element of $\langle I_{i,j}
\rangle$. By $\F_q$-linearity of $\mathfrak{C}_i$ and
$\psi_{i,j}$'s, it is clear that the concatenation is an additive
subgroup of $R'$ which is closed under scalar multiplication by
elements of $\F_q$. Note that for a nonzero coordinate $c_j$ of a
codeword in $\mathfrak{C}_i$, $\psi_{i,j}$ identifies $\alpha_i c_j
\in \E_{i,j}=\E_i$ with $x\psi_{i,j}(c_j)\in I_{i,j}$, since it is a
field isomorphism between $\E_i$ and $\langle I_{i,j} \rangle
=\langle \theta_{i,j} \rangle$ in this case (see (\ref{isoms-2}) and
the discussion following it). Therefore we have
\begin{eqnarray*}
x\cdot \left( \psi_{i,0}(c_0),\ldots,
\psi_{i,\ell-1}(c_{\ell-1})\right) & = & \left( \psi_{i,0}(\alpha_i
c_0),\ldots,
\psi_{i,\ell-1}(\alpha_i c_{\ell-1})\right)\\
&=& \alpha_i \left( \psi_{i,0}(c_0),\ldots,
\psi_{i,\ell-1}(c_{\ell-1})\right)
\end{eqnarray*}
Since $\mathfrak{C}_i$ is an $\E_i$-linear code, $\alpha_i
(c_0,\ldots ,c_{\ell -1})$ is also a codeword of $\mathfrak{C}_i$.
Hence,
$$x\cdot \left( \psi_{i,0}(c_0),\ldots,
\psi_{i,\ell-1}(c_{\ell-1})\right)=x\cdot \Psi_i\left( c_0,\ldots
,c_{\ell -1} \right)$$ is also a codeword of $\langle I_i \rangle
\Box \mathfrak{C}_i$ and this concatenation is an
$\F_q[x]$-submodule of $R'$. If we take the direct sum of several
such concatenations, the result is again an submodule of $R'$, i.e.
a GQC code.
\end{proof}

\begin{rem} \label{const-outer}
Note from the proof of Theorem \ref{generalized Jensen} (i) that the
outer codes of the GQC code $C$ are of the form (for each $1\leq i
\leq s$):
\begin{eqnarray*}
\mathfrak{C}_i & = & \left\{\left(\varphi_{i,0}\left(c^0(x)I_{i,0}(x)\right),\ldots ,\varphi_{i,\ell -1}\left(c^{\ell -1}(x)I_{i,\ell -1}(x)\right)\right): \  \left(c^0(x),\ldots ,c^{\ell -1}(x) \right)\in C \right\} \\
& = & \left\{\left(\varphi_{i,0}\left(c^0(x)\right),\ldots
,\varphi_{i,\ell -1}\left(c^{\ell -1}(x)\right)\right): \
\left(c^0(x),\ldots ,c^{\ell -1}(x) \right)\in C \right\},
\end{eqnarray*}
where the last equality follows from
$\varphi_{i,j}(I_{i,j}(x))=1_{i,j}$. The outer code $\mathfrak{C}_i$
is nothing but the constituent $C_i$ of $C$ (Proposition
\ref{GQCconstituents}). Hence, the analogous result for QC codes
extends to GQC codes.
\end{rem}

As in Theorem \ref{traces}, we can obtain a trace representation for
the codewords of a given GQC code, which is straightforward by using
the isomorphism (concatenation map) in (\ref{isoms-4}).

\begin{thm} \label{traces-2}
Consider the $q$-ary GQC code $C$ of length $m_0+\cdots + m_{\ell
-1}$ with the constituents $C=C_{1}\oplus \cdots \oplus C_{s},$
where $C_i \subset \E_{i}^\ell$ is linear over $\E_{i}$ of length
$\ell$ for each $1\leq i \leq s$. Assume that each $m_j$ is
relatively prime to $q$ and let $\alpha_1,\ldots , \alpha_s$ be
fixed roots of the polynomials $f_1, \ldots ,f_s$, describing the
fields $\E_1, \ldots, \E_s$. Then an arbitrary codeword $c\in C$ has
the form
$$
c=\left(c_0(\lambda_1,\ldots , \lambda_s) \mid  c_1(\lambda_1,\ldots
, \lambda_s) \mid \cdots \mid c_{\ell -1}(\lambda_1,\ldots ,
\lambda_s) \right),
$$
where $\lambda_i=(\lambda_{i,0},\ldots ,\lambda_{i,\ell-1})$ is a
codeword in $C_i$, for each $i=1,\ldots,s$, and for $j\in \{0,\ldots
, \ell -1\}$, the $j^{th}$ column has length $m_j$ and it is of the
form
$$c_j(\lambda_1,\ldots , \lambda_s)=\frac{1}{m_j}\left(\sum\limits_{i=1}^{s}\Tr_{\E_{i}/\F_q}\left(\lambda_{i,j}\alpha_i^{-k_j}\right) \right)_{0\leq k_j \leq m_j-1} .$$
\end{thm}

\begin{rem}
Note that for $m_0=\cdots=m_{\ell-1}$, this coincides with the trace
representation of a length $m\ell$ QC code (cf. Theorem
\ref{traces}). However, the trace representation in Theorem
\ref{traces} describes codewords by their rows whereas Theorem
\ref{traces-2} provides a column-wise description of codewords in a
GQC code.
\end{rem}


\begin{ex}
Let $m_0=3$, $m_1=5$ and $q=2$. We will consider a binary GQC code
$C$ of length $3+5=8$. We have
$$R'=R_0 \times R_1 = \F_2[x]/\langle x^3-1\rangle \times
\F_2[x]/\langle x^5-1\rangle $$ and
\begin{equation} \label{irred facts}
x^3-1 = (x+1)(x^2+x+1) \ \ \ \ \ x^5-1 = (x+1)(x^4+x^3+x^2+x+1).
\end{equation}
Therefore, $s=3$ and $R'$ decomposes as follows:
\begin{eqnarray*}
R' &\cong& \left(\F_2[x]/\langle x+1\rangle \times \F_2[x]/\langle x+1\rangle\right)\\
 &\oplus& \left(\F_2[x]/\langle x^2+x+1\rangle \times \{0\} \right)\\
 &\oplus& \left(\{0\} \times \F_2[x]/\langle x^4+x^3+x^2+x+1\rangle\right).
\end{eqnarray*}
Let $1,\xi_1,\xi_2$ be the fixed roots of the irreducible factors in
(\ref{irred facts}), respectively. Let $C_1 \subseteq \F_2^2$, $C_2
\subseteq \F_4^2$, $C_3 \subseteq \F_{16}^2$ be the constituents of
$C$. Note that the second (first) coordinate of every codeword in
$C_2$ (in $C_3$) must be zero due to the decomposition $R'$ above.
We write $\Tr_{\F_{16}/\F_2}(\alpha)=\Tr(\alpha)$ as short. Then, by
Theorem \ref{traces-2}, the codewords of $C$ are of the form (cf.
Theorem 6.7 and 6.14 in \cite{LS})
$$(z_1+2a-b|z_1-a+2b|z_1-a-b|z_2+\Tr(y)|z_2+\Tr(y\xi_2^{-1})|z_2+\Tr(y\xi_2^{-2})|z_2+\Tr(y\xi_2^{-3})|z_2+\Tr(y\xi_2^{-4})),$$
where $(z_1,z_2) \in C_1$, $a+\xi_1b \in C_2$ ($a,b \in \F_2$) and
$y \in C_3$.

Moreover, we can simplify this expression further, by using the fact
$2a=2b=0$ in $\F_2$ and by setting $y= c+\xi_2 d+\xi_2^2 e+\xi_2^3
f$, for some $c,d,e,f \in \F_2$, as follows:
$$(z_1+b|z_1+a|z_1+a+b|z_2+d+e+f|z_2+c+e+f|z_2+c+d+f|z_2+c+d+e|z_2+c+d+e+f),$$
where $(z_1,z_2) \in C_1$, $a+\xi_1b \in C_2$ and $c+\xi_2 d+\xi_2^2
e+\xi_2^3 f \in C_3$.
\end{ex}

\section{Multilevel Concatenated View of GQC Codes and a Minimum Distance Bound} \label{multilevel}
A direct sum of concatenated codes can be seen as a multilevel
(generalized) concatenation. Linear generalized concatenated codes
were introduced by Blokh and Zyablov (\cite{BZ}), which enabled
Jensen to obtain a minimum distance bound for QC codes
(\cite[Theorem 4]{J}; see also \cite[Theorem 3.3]{GO}). We refer to
Section 2 in Dumer's chapter \cite{D} for more information on
multilevel concatenation.

Note however that for a QC code, or multilevel concatenations as
described in \cite[Section 2]{D}, symbols in the codewords of outer
codes are mapped to inner codes of the same length (length $m$ in
Theorem \ref{Jensen's thm}), whereas this is not the case for the
concatenated structure of GQC codes ($m_0,\ldots , m_{\ell -1}$ in
Theorem \ref{generalized Jensen}). Our goal is to adapt the
multilevel approach to GQC codes and obtain a minimum distance bound
as Jensen did for QC codes. We will first define multilevel
concatenation in a setting applicable to GQC codes. We continue with
the notation introduced so far.

Let $C$ be a $q$-ary GQC code of length $m_0+\cdots + m_{\ell -1}$
with the outer codes (or constituents) $C_1,\ldots ,C_s$. Recall
that $C_i\subset \E_{i,0}\times \cdots \times \E_{i,\ell -1}$ is an
$\E_i$-linear code of length $\ell$ for each $i$. Consider the
following set:

$$B:=\left\{\left(\begin{array}{ccc} c_{1,0} & \ldots & c_{1,\ell -1} \\ \vdots & \vdots & \vdots \\ c_{s,0} & \ldots & c_{s,\ell -1} \end{array} \right): \ (c_{i,0}, \ldots , c_{i,\ell -1})\in C_i \ \mbox{for $1\leq i \leq s$} \right\}.$$
We can view $B$ as a length $\ell$ code over a mixed alphabet
$(\E_{1,0}\times \cdots \times \E_{s,0}) \times \cdots \times
(\E_{1,\ell -1}\times \cdots \times \E_{s,\ell -1})$, which is
$\F_q$-linear with $|B|=\displaystyle{\prod_{i=1}^s |C_i|}$. We note
that $B$ will be the outer code in the multilevel concatenation
scheme.

For each $j=0,\ldots , \ell -1$, we use the maps $\psi_{i,j}$'s in
(\ref{isoms-2}) to define the following $\F_q$-linear isomorphisms:
\begin{equation}\label{con-1}
\begin{array}{lllll}
\psi_j & : & \E_{1,j}\times \cdots \times \E_{s,j} & \rightarrow & \langle I_{1,j}\rangle \oplus \cdots \oplus \langle I_{s,j}\rangle \subset R_j \\
& & (a_{1,j},\ldots ,a_{s,j}) & \mapsto & \psi_{1,j}(a_{1,j})+\cdots
+ \psi_{s,j}(a_{s,j})
\end{array}
\end{equation}
The multilevel concatenated code is defined as
\begin{equation}\label{con-2}
\psi(B):=\left\{ \Bigl(\psi_0\left(c_{1,0},\ldots
,c_{s,0}\right),\ldots , \psi_{\ell -1}\left(c_{1,\ell -1},\ldots
,c_{s,\ell -1}\Bigr)  \right):  \ \left(\begin{array}{ccc} c_{1,0} &
\ldots & c_{1,\ell -1} \\ \vdots & \vdots & \vdots \\ c_{s,0} &
\ldots & c_{s,\ell -1} \end{array} \right) \in B \right\}.
\end{equation}
Observe that the maps $\psi_0,\ldots ,\psi_{\ell -1}$ concatenate
each symbol in the codewords of $B$, which comes from mixed
cross-product alphabets as described above, to length $m_0,\ldots
,m_{\ell -1}$ words respectively. It is also clear that $\dim_{\F_q}
\psi(B)=\displaystyle{\sum_{i=1}^s \dim_{\F_q} C_i}=\dim_{\F_q} C$.

\begin{prop}\label{multi-view}
$$\psi(B)=\bigoplus_{i=1}^s \langle I_i \rangle \Box C_i .$$
\end{prop}

\begin{proof}
A codeword in $\psi(B)$ is of the form
$$\Bigl(\bigl(\psi_{1,0}(c_{1,0})+\cdots + \psi_{s,0}(c_{s,0})\bigr), \ldots , \bigl(\psi_{1,\ell -1}(c_{1,\ell -1})+\cdots + \psi_{s,\ell -1}(c_{s,\ell -1})\bigr) \Bigr),$$
which can be rewritten as
$$\bigl(\psi_{1,0}(c_{1,0}), \ldots , \psi_{1,\ell -1}(c_{1,\ell -1})\bigr) + \cdots + \bigl(\psi_{s,0}(c_{s,0}), \ldots , \psi_{s,\ell -1}(c_{s,\ell -1})\bigr).$$
This expression also belongs to $\displaystyle{\bigoplus_{i=1}^s}
\langle I_i \rangle \Box C_i$ (cf. Remark \ref{concat-defn}), hence
$\psi(B)\subseteq \displaystyle{\bigoplus_{i=1}^s} \langle I_i
\rangle \Box C_i$. The result follows since both codes have the same
$\F_q$-dimension.
\end{proof}

So, we obtained another way of presenting the GQC code $C$. The
advantage of this is that it makes it possible to prove the minimum
distance bound on GQC codes

\begin{thm}\label{dist bound GQC}
Let $C$ be a GQC code with nonzero constituents $C_{i_1},\ldots
,C_{i_g}$, where $\{i_1,\ldots ,i_g\}\subseteq \{1,\ldots ,s\}$. Let
$d_u$ denote the minimum distance of $C_{i_u}$, for each $1\leq u
\leq g$ and assume that $d_1\leq d_2\leq \cdots \leq d_g$. If we set
$$D_u:= \tiny{\displaystyle{\min_{\begin{array}{c} J\subset \{0,1,\ldots , \ell -1\} \\ |J|=d_u \end{array}}}} \left\{\sum_{t\in J} d\Bigl(\langle I_{i_1,t}\rangle \oplus \langle I_{i_2,t}\rangle \oplus \cdots \oplus \langle I_{i_u,t}\rangle \Bigr) \right\} $$
for $1\leq u \leq g$, then
$$d(C)\geq \min\{D_1,D_2,\ldots , D_g\}.$$
\end{thm}

\begin{proof}
Codewords in $B$ have $g$ rows coming from the constituents of $C$.
For any $u \in \{1,\ldots ,g\}$, consider a codeword $b\in B$ whose
first $u$ rows are nonzero codewords from the corresponding
constituents and the remaining rows are the zero codewords. Let us
denote the columns (symbols in the mixed alphabets) of $b$ by
$(b_0,\ldots ,b_{\ell-1})$. By assumption on the ordering of minimum
distances of the constituents, $b$ has at least $d_u$ nonzero
columns. By linearity of $\Psi$, a zero (nonzero) column in $b$ is
mapped to the zero (nonzero) codeword in the corresponding image.
Again due to linearity, zero entries in nonzero columns (e.g. the
last $g-u$ entry in each nonzero column) are also mapped to zeros in
the image. Therefore, if $0\leq t_1,\ldots , t_{d_u} \leq \ell -1$
denotes nonzero columns of $b$, then
$\Psi(b)=\bigl(\psi_0(b_0),\ldots , \psi_{\ell -1}(b_{\ell
-1})\bigr)$ lies in
$$\Bigl(\langle I_{i_1,t_1}\rangle \oplus \cdots \oplus \langle I_{i_u,t_1}\rangle \Bigr) \times \cdots \times  \Bigl(\langle I_{i_1,t_{d_u}}\rangle \oplus \cdots \oplus \langle I_{i_u,t_{d_u}}\rangle \Bigr) \  \mbox{(cf. (\ref{con-1}) and (\ref{con-2}))}.$$
Hence the weight of $\Psi(b)$ is at least
$$\sum_{k=1}^{d_u} d\Bigl(\langle I_{i_1,t_k}\rangle \oplus \cdots \oplus \langle I_{i_u,t_k}\rangle \Bigr).$$
If we consider all possible choices of $d_u$ nonzero columns for
$b\in B$ as above, codewords obtained this way in the image of
$\Psi$ (i.e. $C$) have weights greater than or equal to $D_u$.
Applying the same argument with each $u=1,\ldots ,g$, we see that
codewords of $C$ arising this way from $B$ have weights at least
$D:=\min\{D_1,D_2,\ldots , D_u\}$.

Now suppose $c=\Psi(b)$ is a codeword in $C$, where $b\in B$ has
different configuration of nonzero rows, $\mu_1 < \mu_2 < \cdots
<\mu_{e} \in \{1,\ldots , g\}$. Arguing as above, for some subset
$J$ of $\{0,1,\ldots , \ell -1\}$ of cardinality $|J|=d_{\mu_e}$,
the weight $w(c)$ of such $c$ is at least
$$\sum_{t\in J} d\Bigl(\langle I_{i_{\mu_1},t}\rangle \oplus \langle I_{i_{\mu_2},t}\rangle \oplus \cdots \oplus \langle I_{i_{\mu_e},t}\rangle \Bigr).$$
For each $t\in J$ we have
$$\Bigl( \langle I_{i_{\mu_1},t}\rangle \oplus \langle I_{i_{\mu_2},t}\rangle \oplus \cdots \oplus \langle I_{i_{\mu_e},t}\rangle \Bigr) \subset \Bigl( \langle I_{i_{1},t}\rangle \oplus \langle I_{i_{2},t}\rangle \oplus \cdots \oplus \langle I_{i_{\mu_e},t}\rangle\Bigr).$$
Hence $w(c)\geq D_{\mu_e} \geq D$. Therefore $D$ is a lower bound
for the weights of all codewords in $C$.
\end{proof}

\begin{rem}
Suppose $C$ is a QC code with nonzero constituents $C_{i_1},\ldots
,C_{i_g}$, whose minimum distances are ordered as in Theorem
\ref{dist bound GQC}. If $C$ is of length $m\ell$ and index $\ell$,
then $m_0=\cdots =m_{\ell -1}=m$ and $\langle I_{i_u,t}\rangle
=\langle \theta_{i_u}\rangle$ for any $t\in \{0,\ldots , \ell -1\}$
and any $u\in \{1,\ldots , g\}$. (cf. Section \ref{intro}). Then for
any $J\subset \{0,1,\ldots , \ell -1\}$ with $|J|=d(C_{i_u})$, we
have
\begin{eqnarray*}
\sum_{t\in J} d\Bigl(\langle I_{i_1,t}\rangle \oplus \langle I_{i_2,t}\rangle \oplus \cdots \oplus \langle I_{i_u,t}\rangle \Bigr) & = & \sum_{t\in J} d\Bigl(\langle \theta_{i_1}\rangle \oplus \langle \theta_{i_2}\rangle \oplus \cdots \oplus \langle \theta_{i_u}\rangle \Bigr) \\
& = & d(C_{i_u})d\Bigl(\langle \theta_{i_1}\rangle \oplus \langle
\theta_{i_2}\rangle \oplus \cdots \oplus \langle \theta_{i_u}\rangle
\Bigr).
\end{eqnarray*}
Hence the bound in Theorem \ref{dist bound GQC} takes the form
$$d(C)\geq \displaystyle{\min_{1\leq u \leq g}} \left\{d(C_{i_u})d\Bigl(\langle \theta_{i_1}\rangle \oplus \langle \theta_{i_2}\rangle \oplus \cdots \oplus \langle \theta_{i_u}\rangle \Bigr) \right\},$$
for a QC code $C$, which is exactly Jensen's bound (see
\cite[Theorem 4]{J}, \cite[Theorem 3.3]{GO}).
\end{rem}

\begin{rem}
Esmaeili and Yari also found a minimum distance bound for GQC codes
but their bound only applies to one-generator GQC codes
(\cite[Theorem 4]{EY}).
\end{rem}

\section{Self-Dual and LCD Cases} \label{sd-lcd section}

Let us write the factorization  of the polynomials $x^{m_j}-1$ (for
all $0\leq j \leq \ell -1$) into irreducible polynomials in
$\F_q[x]$ as follows, which is needed for dual code analysis (see
also \cite{LS}):
\begin{equation}\label{irreducibles-3}
x^{m_j}-1=g_{1}(x)^{v_{1,j}}\cdots
g_{r}(x)^{v_{r,j}}h_{1}(x)^{w_{1,j}}h_{1}^*(x)^{w_{1,j}}\cdots
h_{p}(x)^{w_{p,j}}h_{p}^*(x)^{w_{p,j}},
\end{equation}
Here, $g_{i}$'s are self-reciprocal, $h_{t}^*$ denotes the
reciprocal of $h_{t}$ and $v_{i,j}, w_{t,j} \in \{0,1\}$, for all
$i,j,t$.

Let $\G_{i}=\F_q[x]/\langle g_{i}(x) \rangle$,
$\HH_{t}'=\F_q[x]/\langle h_{t}(x) \rangle$ and
$\HH_{t}''=\F_q[x]/\langle h_{t}^*(x) \rangle$ for each $i$ and $t$.
Let us define (cf.(\ref{extensions}))
\begin{equation}\label{extensions-2}
\begin{array}{ccc}
  \G_{i,j}= \begin{cases}
   \ \G_i,   \text{ if } v_{i,j}=1, \\
   \{0\},   \text{ if } v_{i,j}=0.
  \end{cases} &  \HH_{t,j}'= \begin{cases}
   \ \HH_{t}',   \text{ if } w_{t,j}=1, \\
   \{0\},   \text{ if } w_{t,j}=0.
  \end{cases} & \HH_{t,j}''= \begin{cases}
   \ \HH_{t}'',   \text{ if } w_{t,j}=1, \\
   \{0\},   \text{ if } w_{t,j}=0.
  \end{cases}
\end{array}
\end{equation}
By Chinese Remainder Theorem, decomposition of $R_j$ in
(\ref{CRT-4}) now becomes:
\begin{equation} \label{CRT-8}
R_j \cong \left( \bigoplus_{i=1}^{r} \G_{i,j} \right) \oplus \left( \bigoplus_{t=1}^{p} \Bigl( \HH_{t,j}' \oplus \HH_{t,j}'' \Bigr) \right), \ \mbox{for $0\leq j \leq \ell-1$.}\\
\end{equation}
Hence, we get the following isomorphism for $R'= R_0 \times \cdots
\times R_{\ell-1}$:
\begin{equation} \label{CRT-9} \small
R'\cong \left(\bigoplus_{i=1}^{r} (\G_{i,0} \times \cdots \times
\G_{i,\ell-1})\right) \oplus \left(\bigoplus_{t=1}^{p} (\HH_{t,0}'
\times \cdots \times \HH_{t,\ell-1}')\right) \oplus
\left(\bigoplus_{t=1}^{p} (\HH_{t,0}'' \times \cdots \times
\HH_{t,\ell-1}'')\right),
\end{equation} which implies
\begin{equation*}
R'\subseteq \left(\bigoplus_{i=1}^{r} \G_i^{\ell}\right) \oplus
\left(\bigoplus_{t=1}^{p} (\HH'_{t})^{\ell}\right) \oplus
\left(\bigoplus_{t=1}^{p} (\HH''_{t})^{\ell}\right),
\end{equation*}
since for each $j$, $\G_{i,j} \subset \G_i$, $\HH'_{t,j} \subset
\HH'_{t}$ and $\HH''_{t,j} \subset \HH''_{t}$ by
(\ref{extensions-2}).

Hence, a GQC code $C\subset R'= R_0 \times \cdots \times R_{\ell-1}$
viewed as an $\Fq[x]$-submodule of $R'$ now decomposes as (cf.
(\ref{constituents-2}))
\begin{equation} \label{constituents-3}
C=\left( \bigoplus_{i=1}^{r} C_i \right) \oplus \left(
\bigoplus_{t=1}^{p} \Bigl( C_{t}' \oplus C_{t}'' \Bigr) \right).
\end{equation}
where $C_i$'s are the $\G_i$-linear constituents of $C$ of length
$\ell$, for all $i=1,\ldots,r$, $C'_{t}$'s and $C''_{t}$'s are the
$\HH'_t$-linear and $\HH''_t$-linear constituents of $C$ of length
$\ell$, respectively, for all $t=1,\ldots,p$. By fixing roots
corresponding to irreducible factors $g_i,h_j,h_j^*$ and via Chinese
Remainder Theorem (cf. Section 2), one can write explicitly the
constituents as in (\ref{explicit constituents-2}) in this setting
as well.

It is clear that the cardinality of each $\G_i$, say $q_i$, is an
even power of $q$. Each $\G_i^{\ell}$ is equipped with the Hermitian
inner product, which is defined for $\vec{c}=(c_{i,0},\ldots
,c_{i,\ell-1}), \vec{d}=(d_{i,0},\ldots ,d_{i,\ell-1})\in
\G_i^{\ell}$ as
\begin{equation}\label{hermprod}
\langle \vec{c},\vec{d}\rangle := \sum_{j=0}^{\ell-1} c_{i,j}
d_{i,j}^{\sqrt{q_i}}.
\end{equation}
For $1\leq t \leq p$, $\HH_t'^{\ell}$ and $\HH_t''^{\ell}$ are
equipped with the usual Euclidean inner product.

The dual of a GQC code is also GQC. The proof of the following
result will be omitted, since it follows the same lines of the
analogous result given for QC codes in \cite{LS2}.
\begin{prop}\label{duality}
Let $C$ be a GQC code with CRT decomposition as in
(\ref{constituents-3}). Then its dual code $C^{\bot}$ is of the form
\begin{equation}\label{dual}
C^\bot=\left( \bigoplus_{i=1}^r C_i^{\bot_h} \right) \oplus \left(
\bigoplus_{t=1}^p \Bigl( C_t''^{\bot_e} \oplus C_t'^{\bot_e} \Bigr)
\right),
\end{equation}
where  $\bot_h$ denotes the Hermitian dual on $G_i^\ell$ (for all
$1\leq i \leq s$) and  $\bot_e$ denotes the Euclidian dual on
$\HH_i'^{\ell}=\HH_i''^{\ell}$ (for all $1\leq i \leq t$).
\end{prop}

Recall that a linear code $C$ is said to be self dual, if $C=C^\bot$
and $C$ is called linear complementary dual (LCD) if $C \cap C^\bot
= \{0\}$. Let us now characterize self-dual and LCD GQC codes via
their constituents (see \cite{GOS,SLS}).

\begin{thm} \label{SD-CDcriteria}
Let $C$ be a $q$-ary GQC code of length $m_0+\cdots + m_{\ell-1}$,
whose CRT decomposition is as in (\ref{constituents-3}).
\begin{enumerate}
\item $C$ is self-dual if and only if $C_i$ is Hermitian self-dual over $\G_i$, for all $1\leq i \leq r$, and $C_t'' = C_t'^{\bot_e}$ over $\HH_t'=\HH_t''$, for all $1\leq t \leq p$.
\item $C$ is LCD if and only if $C_i$ is Hermitian LCD over $\G_i$, for all $1\leq i \leq r$, and $C_t' \cap C_t''^{\bot_e}=\{0\}$, $C_t'' \cap C_t'^{\bot_e}=\{0\}$ over $\HH_t'=\HH_t''$, for all $1\leq t \leq p$.
\end{enumerate}
\end{thm}

\begin{proof}
Immediate from the CRT decompositions of $C$ in
(\ref{constituents-3}) and of its dual $C^{\bot}$ in (\ref{dual}).
\end{proof}
The following special cases are easy to derive from Theorem
\ref{SD-CDcriteria} above.
\begin{cor} \label{CDinstance}
\begin{enumerate}
\item If the CRT decomposition of $C$ is as in (\ref{constituents-3}) with Hermitian self-dual codes $C_i$ over $\G_i$, for all $1\leq i \leq r$, and $C_t'=C_t''=\{0\}$ over $\HH_t'=\HH_t''$, for all $1\leq t \leq p$, then $C$ is self-dual.
\item If the CRT decomposition of $C$ is as in (\ref{constituents-3}) with Euclidean LCD codes $C_t'=C_t''$ over $\HH_t'=\HH_t''$, for all $1\leq t \leq p$ and Hermitian LCD codes $C_i$ over $\G_i$, for all $1\leq i \leq r$, then $C$ is LCD.
\end{enumerate}
\end{cor}

\section{Asymptotics} \label{asymptotics section}
The existence of the asymptotically good self-dual GQC codes is
shown in \cite{SLS}. In this section, we will analyze the asymptotic
performance of the complementary dual GQC codes, which are
constructed by using asymptotically good QC complementary dual
(QCCD) codes (see \cite{GOS}). We need the following results.

\begin{lem} \cite[Proposition 10]{CG}\label{LCD lemma}
Suppose that $C_1$ and $C_2$ are two $q$-ary LCD codes with
parameters $[n_1,k_1,d_1]$ and $[n_2,k_2,d_2]$ respectively. Let
$E=[C_1|C_2]:=\{[u|v] : u\in C_1, v \in C_2\}.$ Then $E$ satisfies
the following:
\begin{itemize}
\item[(i)] $E \subseteq \F_q^{n_1+n_2}$ is an LCD code.
\item[(ii)] $\dim E=k_1+k_2$.
\item[(iii)] $d(E)=\min\{d_1, d_2\}.$
\end{itemize}
\end{lem}

Note that for more than two codes over the same alphabet, say
$C_1,\ldots ,C_a$, one can similarly define $[C_1|\cdots |C_a]$ and
the parameters are also determined similarly. Moreover, it is also
clear that if each $C_i$ is LCD, then the same also holds for
$[C_1|\cdots |C_a]$.

\begin{thm}\cite[Corollary 3.8]{GOS}\label{asymptotic-QCCD}
For any pair $q$ and $m$, which are relatively prime, there exists
an asymptotically good sequence of QCCD codes over $\F_q$ where each
QC code in the sequence has index length/$m$.
\end{thm}

The construction in Lemma \ref{LCD lemma} with QC component codes
yields a GQC code.

\begin{lem} \label{LCD lemma-2}
Suppose that $C_i$ is a QC code of length $m_i\ell_i$ and index
$\ell_i$ for each $1\leq i \leq a$, where $m_i$'s are pairwise
distinct. Then $[C_1|\cdots |C_a]$ is a GQC code of block lengths
$(\underbrace{m_1,\ldots ,m_1}_{\ell_1},\ldots ,
\underbrace{m_a,\ldots , m_a}_{\ell_a})$.
\end{lem}

\begin{proof}
Let $R_i=\F_q[x]/\langle x^{m_i}-1\rangle$ for each $i$ and recall
that $C_i$ is an $R_i$-submodule in $R_i^{\ell_i}$. In polynomial
representation, codewords are of the form $[u_1(x)|\cdots |u_a(x)]$,
where
$$u_i(x)=(u_{i,0}(x),\ldots , u_{i,\ell_i-1}(x)) \in C_i \subset R_i^{\ell_i}.$$
It is enough to show that $[C_1|\cdots |C_a]$ is closed under
multiplication by $x$ in $R_1^{\ell_1} \times \cdots \times
R_a^{\ell_a}$, which is true since each $C_i$ is a QC code and hence
closed under multiplication by $x$. The claim about the block
lengths of the resulting GQC code is also clear.
\end{proof}

Let $m_1$ and $m_2$ be distinct positive integers coprime to $q$.
Take two asymptotically good sequences of QCCD codes over $\F_q$,
say $(C_i)_{i \geq 1}$ and $(D_i)_{i \geq 1}$, with parameters
$[m_1\ell_i, k_i, d_i]$ and $[m_2\ell_i, k'_i, d'_i]$ for members of these sequences, respectively. We have\\
\begin{minipage}[r]{0.45\textwidth}
\begin{eqnarray}
R_{C_i} &=& \lim_{i \to \infty}
\dfrac{k_i}{m_1\ell_i} > 0, \nonumber\\
\delta_{C_i} &=& \lim_{i \to \infty} \dfrac{d_i}{m_1\ell_i} > 0,
\end{eqnarray}
\end{minipage} and
\begin{minipage}[l]{0.45\textwidth}
\begin{eqnarray*}
R_{D_i} &=& \lim_{i \to \infty}
\dfrac{k'_i}{m_2\ell_i} > 0, \\
\delta_{D_i} &=& \lim_{i \to \infty} \dfrac{d'_i}{m_2\ell_i} > 0.
\end{eqnarray*}
\end{minipage}\vspace{0.4cm}

Note that the existence of such sequences is guaranteed by Theorem
\ref{asymptotic-QCCD}. For each $i \geq 1$, set $ E_i = [C_i|D_i]$.
Then by Lemmas \ref{LCD lemma} and \ref{LCD lemma-2}, $E_i$ is a
GQCCD code of length $(m_1+m_2)\ell_i$, and block lengths
$(\underbrace{m_1,\ldots ,m_1}_{\ell_i}, \underbrace{m_2,\ldots ,
m_2}_{\ell_i})$. We have
\begin{eqnarray}
R_{E_i} &=& \lim_{i \to \infty}
\dfrac{k_i + k'_i}{(m_1+m_2)\ell_i} = \dfrac{m_1}{(m_1+m_2)}R_{C_i}+\dfrac{m_2}{(m_1+m_2)}R_{D_i} > 0, \nonumber\\
\delta_{E_i} &=& \lim_{i \to \infty} \dfrac{\min
\{d_i,d'_i\}}{(m_1+m_2)\ell_i} =
\min\left\{\dfrac{m_1}{(m_1+m_2)}\delta_{C_i},\dfrac{m_2}{(m_1+m_2)}\delta_{D_i}\right\}
 > 0.
\end{eqnarray}
Hence, we obtain a GQCCD code sequence $(E_i)_{i\geq 1}$ over $\F_q$
where each member of the sequence has block lengths
$(\underbrace{m_1,\ldots ,m_1}_{\ell_i}, \underbrace{m_2,\ldots ,
m_2}_{\ell_i})$. This argument can be generalized from two component
codes to many and the following can be similarly obtained.

\begin{thm} \label{GQCCD-asymp good}
Let $q$ be a prime power and assume that $m_1,m_2,\dots,m_a$ are
pairwise distinct positive integers relatively prime to $q$. Then
there exists an asymptotically good sequence of $q$-ary GQCCD codes
$(E_i)_{i\geq 1}$, where each GQC code in the sequence has block
lengths $(\underbrace{m_1,\ldots ,m_1}_{\ell_i},\ldots ,
\underbrace{m_a,\ldots , m_a}_{\ell_i})$ for some $\ell_i$.
\end{thm}

\section{Conclusion and open problems} \label{conclusion}

We have provided a concatenated structure for GQC codes in the sense
of \cite{GO} and \cite{J}, which gives rise to their trace
representation covering the trace representations of QC and cyclic
subclasses. Moreover, a multilevel concatenated view of GQC codes is
introduced, which leads to a minimum distance bound that extends
Jensen's bound for QC codes. By extending the CRT decomposition of
the base ring into self-reciprocal polynomials and reciprocal pairs
of polynomials, as done in \cite{LS} for QC case, we have obtained
criteria for GQC codes to be self-dual or, respectively, LCD. We
have then showed that long GQCCD codes are good. The enumeration of
GQC codes that are LCD should be studied, with the potential
application of deriving a Gilbert-Varshamov bound. Tabulating GQC
codes parameters in modest lengths is also a worthy goal.

\section{Acknowledgment}
\"{O}zbudak and \"{O}zkaya are supported by T\"{U}B\.{I}TAK project
215E200, which is associated with the SECODE project in the scope of
CHIST-ERA Program. Sol\'{e} is supported by the SECODE Project too.
G\"{u}neri is supported partly by T\"{U}B\.{I}TAK 215E200 (SECODE)
and by T\"{U}B\.{I}TAK 114F432 projects. Sa\c{c}\i kara is supported
by T\"{U}B\.{I}TAK project 114F432.


\begin{thebibliography}{99}

\bibitem {BZ} E.L. Blokh and V.V. Zyablov, ``Coding of generalized concatenated codes", \emph{Probl. Inform. Transm.}, vol. 10, 218-222, 1974.

\bibitem {CG} C. Carlet and S. Guilley, ``Complementary dual codes for counter-measures to side-channel attacks", \emph{Adv. in Math. of Comm.}, vol. 10, no 1, 131-150, 2016.

\bibitem {CW} J.T. Cordaro and T.J Wagner, ``Optimum $(n,2)$ codes for small values of channel error probability", \emph{PGIT}, vol. 13, 349-350, 1967.

\bibitem {D} I. Dumer, ``Concatenated codes and their multilevel generalizations", \emph{Handbook of Coding Theory}, North-Holland, Amsterdam, 1911-1988, 1998.

\bibitem {EY} M. Esmaeili and S. Yari, ``Generalized quasi-cyclic codes: structural properties and code construction", \emph{Applicable Algebra in Engineering, Communications and Computing}, vol. 20, 159-173, 2009.

\bibitem {GO} C. G\"{u}neri and F. \"{O}zbudak, ``The concatenated structure of quasi-cyclic codes and an improvement of Jensen's bound", \emph{IEEE Trans. on Inform. Theory}, vol. 59, no. 2, 979-985, 2013.

\bibitem {GOS} C. G\"{u}neri, B. \"{O}zkaya and P. Sol\'{e}, ``Quasi-cyclic complementary dual codes", \emph{Finite Fields Appl.}, vol. 42, 67-80, 2016.

\bibitem {J} J.M. Jensen, ``The concatenated structure of cyclic and abelian codes", \emph{IEEE
Trans. Inform. Theory}, vol. 31, no. 6, 788-793, 1985.

\bibitem{LS} S. Ling and P. Sol\'{e}, ``On the algebraic structure of quasi-cyclic codes I: finite fields", \emph{IEEE Trans. Inform. Theory}, vol. 47, 2751-2760,  2001.

\bibitem{LS2} S. Ling and P. Sol\'{e}, ``On the algebraic structure of quasi-cyclic codes III: generator theory", \emph{IEEE Trans. Inform. Theory}, vol. 51, 2692-2700,  2005.

\bibitem{LS3} S. Ling and P. Sol\'{e}, ``Good self-dual quasi-cyclic codes exist",  \emph{IEEE Trans. Inform. Theory}, vol. 49, 1052-1053, 2003.

\bibitem{vL} J.H. van Lint, \emph{Introduction to Coding Theory}, Springer, (1982).

\bibitem{SK} I. Siap and N. Kulhan, ``The structure of generalized quasi-cyclic codes", \emph{Applied Mathematics E-Notes}, vol. 5, 24-30,  2005.

\bibitem {SLS} M. Shi, Y. Liu and P. Sol\'{e}, ``Good self-dual generalized quasi-cyclic codes exist", Information Proc. Letters, 2017. Available from \url{http://arxiv.org/abs/1601.02437}.

\end{thebibliography}

\end{document}